\documentclass{article}

\usepackage[preprint]{neurips_2018}

\usepackage[utf8]{inputenc} 
\usepackage[T1]{fontenc}    
\usepackage{hyperref}       
\usepackage{url}            
\usepackage{booktabs}       
\usepackage{graphicx} 
\usepackage{amsfonts}       
\usepackage{nicefrac}       
\usepackage{microtype}      
\usepackage{amsmath}
\usepackage{amstext}
\usepackage{algorithm}
\usepackage[noend]{algpseudocode}
\usepackage{amsthm}

\usepackage{xcolor}

\newtheorem{theorem}{Theorem}
\newtheorem{defn}{Definition}

\newtheorem{lemma}{Lemma}

\newtheorem{corollary}{Corollary}

\newtheorem{proposition}{Proposition}

\usepackage{mathtools}

\newcommand{\eps}{\epsilon}

\newcommand{\bbP}{\mathbb{P}}
\newcommand{\bbE}{\mathbb{E}}
\newcommand{\expectation}{\mathbb{E}}

\mathtoolsset{centercolon}
\newcommand{\defeq}{\mathrel{:\mkern-0.25mu=}}
\newcommand{\eqdef}{\mathrel{=\mkern-0.25mu:}}

\title{Scalable Fair Division for `At Most One' Preferences}

\author{%
  Christian Kroer* \\
  Facebook Core Data Science
  \And
   Alexander Peysakhovich* \\
  Facebook AI Research
}

\begin{document}

\maketitle

\begin{abstract}
Allocating multiple scarce items across a set of individuals is an important practical problem. In the case of divisible goods and additive preferences a convex program can be used to find the solution that maximizes Nash welfare (MNW). The MNW solution is equivalent to finding the equilibrium of a market economy (aka. the competitive equilibrium from equal incomes, CEEI) and thus has good properties such as Pareto optimality, envy-freeness, and incentive compatibility in the large. Unfortunately, this equivalence (and nice properties) breaks down for general preference classes. Motivated by real world problems such as course allocation and recommender systems we study the case of additive `at most one' (AMO) preferences - individuals want at most 1 of each item and lotteries are allowed. We show that in this case the MNW solution is still a convex program and importantly is a CEEI solution when the instance gets large but has a `low rank' structure. Thus a polynomial time algorithm can be used to scale CEEI (which is in general PPAD-hard) for AMO preferences. We examine whether the properties guaranteed in the limit hold approximately in finite samples using several real datasets.
\end{abstract}

%
%



%
%


\maketitle

\section{Introduction}
Many important practical problems involve distributing a finite set of items across a set of individuals in a `good' way. While the exact definition of `good' is hotly debated, typically there are several desiderata: the allocation should be efficient in some way, should be fair in some way, and should have good incentive properties for the individuals \citep{brams1996fair,roth2015gets,brandt2016handbook,milgrom2004putting}. In this paper we study allocation in an environment where preferences are additive, individuals want at most one of any given item, monetary transfers are not possible, and lotteries are allowed. Examples of such assignment problems in practice are course assignment at universities \citep{budish2012multi}, pilot-to-plane assignment for airlines, dividing up estates \citep{goldman2014spliddit}, and recommender systems for scarce goods. We show that when instances are large and `low rank' we can construct a strategy-proof mechanism that generates Pareto Optimal and envy-free allocations. This mechanism requires only solving a convex program (potentially via gradient methods) so it is scalable to extremely large assignment problems.

All item assignment mechanisms must trade off between total utility attained and various measures of fairness of the division. One distributionally fair allocation is the one which maximizes the product of the utilities (or the sum of the log utilities). This is called max Nash welfare (MNW \cite{nash1953two,binmore1986nash}). MNW has many other nice properties leading researchers to refer to MNW as an `unreasonably fair' allocation method \citep{caragiannis2016unreasonable}. MNW is employed at spliddit.org which allows individuals to fairly split items such as estates \citep{goldman2014spliddit}. 

Another popular assignment method is competitive equilibrium from equal incomes (CEEI \cite{varian1974efficiency,budish2011combinatorial}). In CEEI individuals are given a budget of faux currency and prices for items are computed such that the total demand for each item (given the budgets) equals the supply. Because CEEI is based on a market equilibrium, we know that CEEI allocations will be Pareto efficient (nobody can be made better off without making someone worse off), envy-free (nobody prefers anyone else's allocation), and strategy-proof in the large (when individuals report their preferences to the CEEI mechanism, they have no incentive to lie). Due to these attractive properties CEEI is employed in the CourseMatch system which is used to assign course slots to individuals at universities \citep{budish2016course}. A downside of the CEEI mechanism is that solving for market equilibrium with general preferences is extremely difficult and even in relatively small instances the CourseMatch system needs to ``use heuristic search that solves billions of Mixed-Integer Programs to output an approximate competitive equilibrium.'' \citep{budish2016course}. Indeed, computation of approximate CEEI in general is known to be PPAD hard \citep{othman2016complexity}.

A surprising coincidence occurs when items are divisible and preferences are additive. Here the MNW solution can be found by solving a convex optimization problem \textit{and} the MNW allocations happens to coincide with a CEEI allocation \citep{eisenberg1959consensus}. This program is often called the Eisenberg-Gale (EG) convex program. This means that MNW solutions inherit the desirable properties above, and also that a CEEI can be computed extremely efficiently.

However, even if there is no complementarity between items the linear allocation rule is restrictive. For example, in the case of course allocation we do not want to assign an individual $5.6$ seats in a course. Thus, in this paper we consider the problem of at most one (AMO) preferences. We assume that preferences for items are additive but that individuals do not desire more than one of an item.\footnote{For ease of notation we consider individuals only wanting at most one of each item, but most of our results carry over to the case where `buyer capacity constraints' hold across items - for example there may be $10$ different brands of vacuum cleaners which individuals may value differently, however still want at most $K$ vacuum cleaners total. In the Appendix we describe this in more detail but keep the AMO constraint in the main text for ease of exposition.} We also assume that lotteries are allowed and so item allocations are now real numbers between $0$ and $1.$

Our main theoretical contribution is to show that when the underlying valuations have a low-rank structure and the instance is large the convex program corresponding to AMO-MNW maximization also yields an allocation which is an approximate CEEI. This can be interpreted either as saying that large scale AMO-MNW has the good properties of CEEI or that approximate large scale linear-AMO-CEEI can be computed in polynomial time.

We then investigate whether good properties of large scale, low-rank AMO-MNW hold in real datasets. We consider two markets, one for course allocation at Harvard Business School \citep{budish2012multi} and one using a recommender system \citep{harper2016movielens}. We show that in these markets AMO-MNW produces good solutions at a great speedup over standard CEEI implementations. Finally, we investigate incentive properties of AMO-MNW by using Bayesian Optimization to search for profitable deviations from truthful revelation. We find that in AMO-MNW deviations are fairly difficult to find and even when found produce only marginal improvements over truthful revelations.
\section{Related Work}
There is a large literature on computing market equilibria in Fisher markets using convex programming \citep{eisenberg1959consensus,shmyrev2009algorithm,cole2017convex} or gradient-based methods \citep{birnbaum2011distributed,nesterov2018computation}. There is also work extending the EG program to new settings such as quasi-linear utilities and indivisible items \citep{cole2017convex,cole2018approximating,caragiannis2016unreasonable}. Our paper complements this existing work as we extend the EG program to the case of AMO preferences and show that when markets are large and low rank this again continues to construct a market equilibrium allocation.

An important real world usage of multi-item allocation is online advertising where ad slots are assigned to individuals via an auction mechanism. This problem is known to be hard in large multi-unit, multi-bidder cases especially with budgets \citep{borgs2005multi}. It is also known that there is an equivalence between paced dynamic auctions and the equilibria resulting from solving the EG problem \citep{conitzer2019pacing}. Our paper shows that EG equilibrium can be used as a flexible tool for such allocation problems with AMO preferences. Note that our results here are for the `fake money' setting but they can be extended to the case where money is real (aka. quasi-linear settings) easily~\citep{cole2017convex}. 

\section{Theory}
We consider the setting of $N$ \text{individuals} with generic element $i$ and $M$ \textbf{items} with generic element $j$. Each item has a finite supply $s_j$ with $s_{max}$ being the maximum supply of an item. Individual $i$ has utility $v_{ij}$ for item $j$. We let $x$ denote an $N \times M$ matrix which we call the \textbf{allocation} of items to individuals with $x_i$ denoting a vector of how much of each item individual $i$ receives.

The main question is what is a good $x$ to use. One answer given by the literature is to maximize the product of utilities - equivalently, the sum of the log utilities. This is known as the max Nash welfare (MNW) solution. It is the solution to the following constrained optimization problem:

\[ 
\max_{x\geq 0} \sum_{i} \text{log} (u_i(x_i)) \text{ s.t. } \forall j \sum_{i} x_{ij} \leq s_j 
\]

When items are divisible and utilities are linear (i.e. when $u_i(x_i) = x_i v_{i} = \sum_{j} x_{ij} v_{ij}$) the problem above is a convex program known as the Eisenberg-Gale (EG) convex program and can be solved efficiently \citep{eisenberg1959consensus}.

Surprisingly, the EG program has an economic interpretation. Starting with the setup above we can consider the case where each individual is given a budget of $1$ unit of faux-currency. Each item is assigned a price $p_j$ and given these prices we define:

\begin{defn} 
Buyer $i$'s \text{demand} is the set of utility maximizing allocations $$D_i (p) = \lbrace x_i \mid v_i x_i \geq v_i x'_i \forall x'_i \text{ and } x_{ij} \leq s_j \rbrace.$$ Note that demand may be set valued but all $x_i \in D_i (p)$ achieve the same utility level. We let this be $\bar{U}_i (p).$
\end{defn}

\begin{defn}
An allocation, price pair ($x^*, p^*$) is a \textbf{market equilibrium} if supply meets demand:
\begin{enumerate} 
\item $x^*_i \in D_i (p^*)$ for all buyers $i$
\item $\sum_{i} x^*_{ij} = s_j$ for all items $j$
\end{enumerate}
\end{defn}

A deep connection between the convex program above and the concept of market equilibrium comes from the following equivalence. Solve the convex program and let  $x^{EG}$  be the solution. Let $\lambda^{EG}$ be the $M$ Lagrange multipliers at the optimum (one per item). It can be shown that $(x^{EG}, \lambda^{EG})$ form a market equilibrium.

This connection tells us a lot about properties of the EG solution. We know that any CEEI is envy-free (no individual prefers another's bundle to their own), is Pareto Optimal (nobody can be made better without someone else being made worse off). In addition, market equilibrium can be used as an allocative mechanism - individuals are asked to report their valuations, an equilibrium is computed, and the resulting allocation is enacted. This mechanism is incentive compatible in the large - i.e. when markets get large, no individual has an incentive to lie about their true preferences.

\subsection{At Most One EG}
\label{sec:amo}
The simple EG solution above is not applicable to many real world problems. This is for two reasons, first, many real world goods are indivisible. However, a fractional allocation can always be thought of as a lottery and resolved at the end of the mechanism. A bigger issue is that in many real world applications such as course allocation, recommender systems, or advertising slots the individuals do not want more than $1$ of any single item. We refer to this as the \textbf{at most one} (AMO) setting.

We can rewrite our slate constrained problem (which we refer to as the AMO Eisenberg-Gale or AMO-EG) as:

\[ 
\max_{x \geq 0} \sum_{i} \text{log} (x_i v_i) \text{ s.t. } \forall j \sum_{i} x_{ij} \leq s_j \text{ and } x_{ij} \leq 1
\]

Unfortunately, the solution to the AMO-EG program no longer has the nice properties of the original EG (although given that it is a maximization problem the resulting allocation remains Pareto optimal). 

To see this, consider the example of $2$ individuals and $2$ items. There is supply $2$ of item $1$ and supply $1$ of item $2$. Both value item $1$ at $1$, individual $1$ values item $2$ at $1$ and individual $2$ values it at $100.$ In the AMO-EG solution here both individuals get $1$ of item $1$ due to the constraint and split item $2$. However, individual $2$ receives more of item $2$ than individual $1$, thus individual $1$ must envy individual $2$. Since the solution exhibits envy, it cannot be a CEEI solution.\footnote{Thanks to Nisarg Shah for providing this example.}

To see where the original EG and CEEI equivalence breaks down, notice that unlike the standard EG problem, the AMO-EG problem now has constraints both at the item level ($M$ supply constraints which we refer to as $\lambda^{AMO}$) and the buyer-item level ($N \times M$ AMO constraints). Due to the presence of buyer-item specific multipliers the $M$ supply-Lagrange multipliers at the optimum may no longer be thought of as prices supporting the solution $x^{AMO}$. 

However, we now give a condition for when they approximately support $x^{AMO}$ as an equilibrium allocation.

\begin{defn}
Let $(x^*, p^*)$ be an allocation, price pair. We say it is a \textbf{$\delta$-approximate CEEI} if for all buyers $i$ we have $$v_i \cdot x_i^* + \delta \geq \bar{U}_i (p^*).$$ In other words, that $x_i^{*}$ is approximately a demand vector for buyer $i$ under the prices.
\end{defn}

We show that when every buyer has a sufficiently-large set of ``twins'' then they receive approximate demands in AMO-EG. Note that in the AMO setting buyers must take into account the AMO constraint when computing demands.

\begin{theorem}\label{maintheorem}
Individual $i$ and $j$ are $\epsilon$-twins if $|| v_i - v_j ||_{\infty} \leq \epsilon.$ For any $\delta>0$ there exists $\epsilon>0$ such that if everyone has $s_{max}+1$ $\epsilon$-twins, then $(x_{AMO}, \lambda_{AMO})$ is a $\delta$-approximate CEEI.
\end{theorem}

The Theorem here is written that for any $\delta$ there exists an $\epsilon.$ In fact, we can derive a much stronger result: given an instantiation of individuals, items, valuations, supplies, and budgets, we can explicitly bound how $\delta$-approximate of a CEEI the AMO-EG solution is by considering the largest $\epsilon$ such that each individual has an $\epsilon$ twin. This bound is quite complex and requires the introduction of a lot of notation, thus we relegate it to the Appendix. In addition, this bound lets us derive a weaker, but harder to interpret condition for $\delta$-approximate CEEI which implies the one in Theorem \ref{maintheorem} above. We provide a more general theorem which implies Theorem \ref{maintheorem} in the Appendix.

Here we give a sketch of the intuition for Theorem \ref{maintheorem}. First, recall that the Lagrange multipliers of a convex program measure the shadow value of locally relaxing a constraint. Consider a group of $s_{max} + 1$ exact twins with identical valuation functions. In any allocation the AMO constraints for some item $j$ bind for at most $s_{max}$ of them. However, the last twin, by assumption, is unconstrained. This means that for that item $j$ the AMO Lagrange multipliers are $0$ at $x^{AMO}$ since the last twin is unconstrained and all twins contribute equally to the convex program's objective. It can then be shown, just as in the standard EG program, that the remaining non-zero Lagrange multipliers (the $M$ supply constraints) can again be thought of as prices. Looking at approximate instead of exact twins is slightly more complex but follows a related logic.

Approximate CEEI is attractive because it approximately retains the nice properties of CEEI. When receiving an item in their demand, a given individual $i$ has no envy (since they can afford any bundle received by some other individual) and receives their proportional share (since prices sum to less than $N$). It follows that in $\delta$-approximate CEEI every individual has at most $\delta$ envy, and receives all but $\delta$ of their proportional share.

\subsection{Low Rank Markets}
We now ask: under what generating process for a large market is the approximate twin condition satisfied?

We propose a simple market generating process that is a model of many real world markets. We look at \textit{low rank markets} where every individual $i$ has a latent vector $\theta_i$ drawn from some set $\Theta$ which is a compact subset of $\mathbb{R}^d_+$. Every item has a vector $\psi$ from some set $\Psi$ which is also some compact subset of $\mathbb{R}^d_+$ with $d$ fixed. The supply of each item is some constant $s>1$. The valuation of an individual is $v_{ij} = \theta_i \cdot \psi_j$. This states that item valuations are not independent and valuation of one item can be at least partially predicted from valuations of other items. This is true in most real world recommendation settings and is exploited in most recommender systems. 

Consider growing a low rank market by drawing $N$ individuals from some continuous distribution on $\Theta$ and drawing $kN$ items from $\Psi.$

\begin{theorem}
  \label{maintheorem probabilistic}
For any $\epsilon, \delta > 0$ there exists $\bar{N}$ such that if $N > \bar{N}$ the probability that $(x_{AMO}, \lambda_{AMO})$ is an $\delta$-approximate CEEI is greater than $1-\delta$.
\end{theorem}

Thus, in large, low-rank markets AMO-EG is a convex program that finds an approximate CEEI.

Low rank markets are a natural model under which our approximate-twin conditions are satisfies. However, there are also many other market models one could construct where these same conditions hold. A commonly used example is \emph{replicator markets}, where we start with some finite base economy, and the market grows by replicating the set of individuals and items. Here every individual gets $k$ exact twins, where $k$ is the number of market replicates. Assuming that the AMO constraint is local to each market replicate this immediately implies the condition of Theorem \ref{maintheorem}.

\subsection{Strategy Proofness for AMO-EG}
An important part of any mechanism is individuals reporting their valuations. If reporting truthfully yields a higher utility than lying a mechanism is known as \textit{strategy proof}. While CEEI is not strategy proof in general, it is known that it becomes \emph{strategy proof in the large} (SP-L). 

We now extend some known results \citep{azevedo2018strategy,peysakhovich2019fair} to our setting to show that AMO-EG is SP-L in low rank markets. Note that the low rank model differs from the results in \cite{azevedo2018strategy} as it assumes an infinite number of outcomes (since allocations are continuous). It also differs from the related SP-L results in \citet{peysakhovich2019fair} as it assumes a growing number of items.

We consider the definition of SP-L introduced by \citet{azevedo2018strategy} and used in \citet{peysakhovich2019fair}. Let $\sigma$ be a mapping (not necessarily truthful) from valuations to reports. Suppose that $\sigma$ has full range (i.e. for any report, there is a type that gives that report with some probability). Consider an individual $i$ who can either report their true valuations $v_i$ or some other valuation $v'_i$ with everyone else reporting according to $\sigma$. Consider $i$'s expected utility with the expectation taken relative to the $n-1$ other individuals behaving according to $\sigma$ with their true types drawn from $\mathcal{F}$.

\begin{defn}
 A mechanism is SP-L if, given $\sigma$ and any $\epsilon > 0$ there exists a $\bar{N}$ such that if $N > \bar{N}$ then the gain from reporting any $v'_i$ instead of their true valuation $v_i$ gains is at most $\epsilon$.
\end{defn}

\begin{theorem}
  \label{thm:spl}
The AMO-EG mechanism is SP-L for low-rank markets.
\end{theorem}

\subsection{Generality of these Results}
In the prior section we have shown results for AMO-EG. We have done this for ease of exposition and because the at-most-one constraint seems reasonable for many practical applications. Importantly, however, our result generalize to a more general market model. In the Appendix we prove the EG, market equilibrium equivalence for a more general model where items are partitioned into item-sets and individuals want at most one item from each item set (for example, in the case of household items there may be multiple vacuum cleaner brands but individuals may still only want one vacuum cleaner). The more general model also allows arbitrary budgets. AMO is a special case where each item set includes just a single item, and CEEI is a special case where all budgets are $1$. For the paper, we stick to the simpler AMO model as it is easier to discuss and gain intuition about.

\section{Experiments}

\subsection{Datasets}
We consider two datasets of allocation preferences. First, we use a synthetic market from a recommender system as used in \citep{kroer2019computing,peysakhovich2019fair}. To construct this market we take the MovieLens1M dataset in which ~6000 users rate subsets of ~4000 movies. We use the same procedure/hyperparamters from \citet{kroer2019computing} to construct estimates for unobserved user/movie pairs. We then take the top 200 users with the most ratings and the top 1000 movies with the most ratings and construct markets from this subset. We consider a growing set of markets by taking $\lbrace 50, 100, 200, 500, 1000 \rbrace$ items and holding the number of buyers fixed. We set the sum of supplies to be in $\lbrace 500, 1000, 2000 \rbrace.$ This procedure gives us 15 nested markets that grow as the conditions of Theorem \ref{maintheorem} suggest.

We also consider an allocation dataset from course allocation at Harvard Business school \citep{budish2012multi}. In this dataset there are $936$ students and $93$ courses with fixed supplies. Students rank their top $30$ courses. To transform these ranks into utilities we assume that a rank $r$ course gives utility $\frac{31-r}{30}$ and courses that receive no ranking have utility $0$ (recall that MNW is invariant to the scale of individual utilities so normalizing the top ranked course to have utility $1$ has no effect on our computation). Note that this dataset is not like the limit markets considered in Theorem \ref{maintheorem} as it has many buyers and relatively few items. 

\subsection{AMO-EG Allocations} 
We now turn to evaluating AMO-EG in our data. We solve for the equilibrium using a commercial convex programming solver that supports maximizing logarithmic utilities (MOSEK 9~\citep{dahl2019primal}), though at larger scales (e.g. millions) we would need to use a first-order method. We refer the reader to \cite{kroer2019computing} for specialized techniques for solving extremely large scale EG problems which are directly generalizable to our AMO-EG setting.

As a first cut, we will compare the EG solution to the AMO-EG solution. In particular, we will see how often the standard EG violates the AMO constraint. If EG solutions do not violate the AMO constraint then AMO-EG is trivially a market equilibrium and indeed can be used out of the box without our modifications. However, figure \ref{eg_violations} shows us this is not the case in any of our datasets. We count both the number of individuals that receive more than one of an item (number of violators) and the amount of items they receive conditional on receiving more than one (violation size).

The naive EG solution allocates almost every buyer in the MovieLens markets more than one of an item and often large multiples. In the course data (right panel) we see that EG allocates some individuals as many as 19 seats in a course. Thus, naive EG ignoring the AMO constraints yields poor solutions for AMO preferences and thus the AMO-EG program is needed.

\begin{figure}[h]
\begin{center}
\includegraphics[scale=.45]{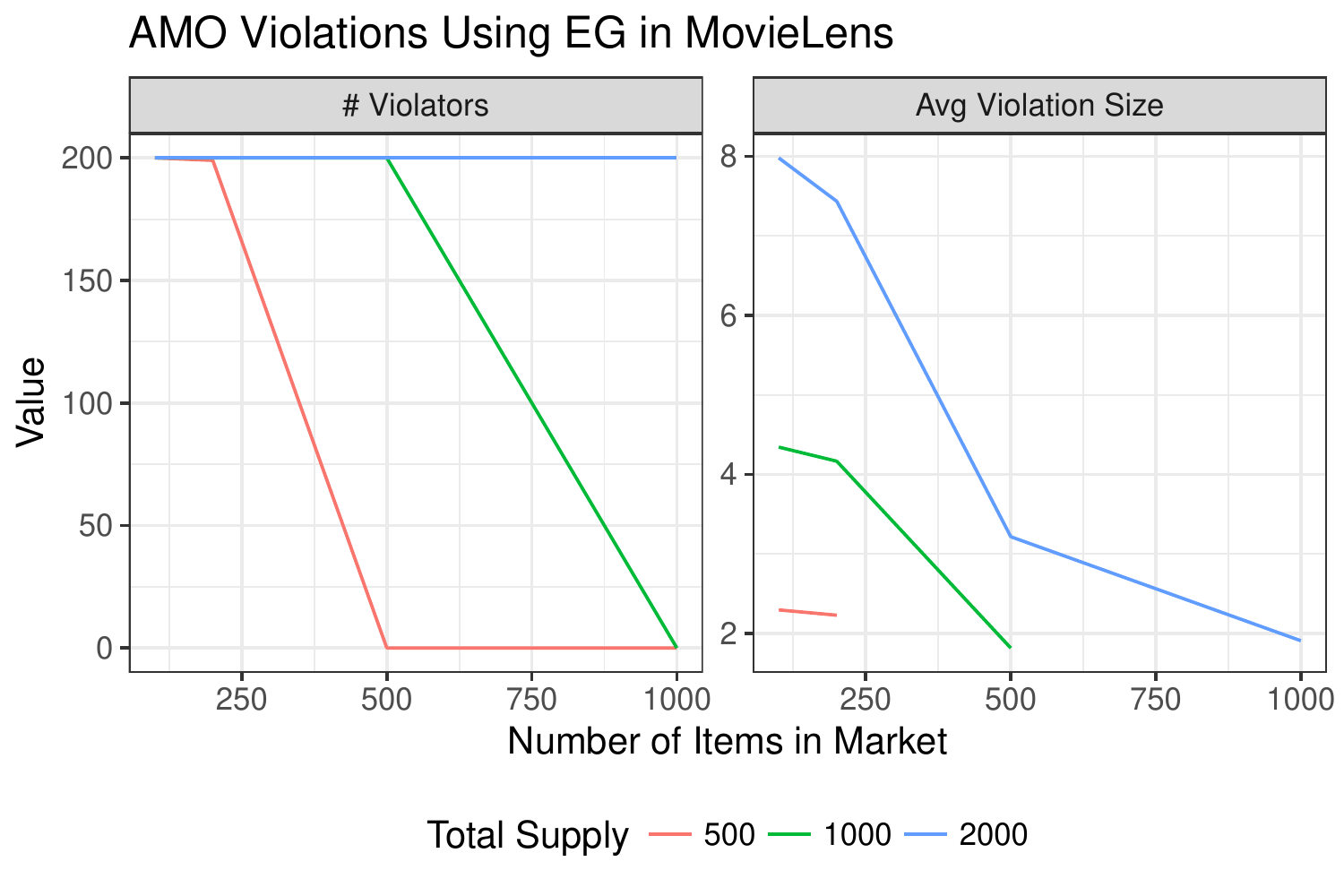}
\includegraphics[scale=.45]{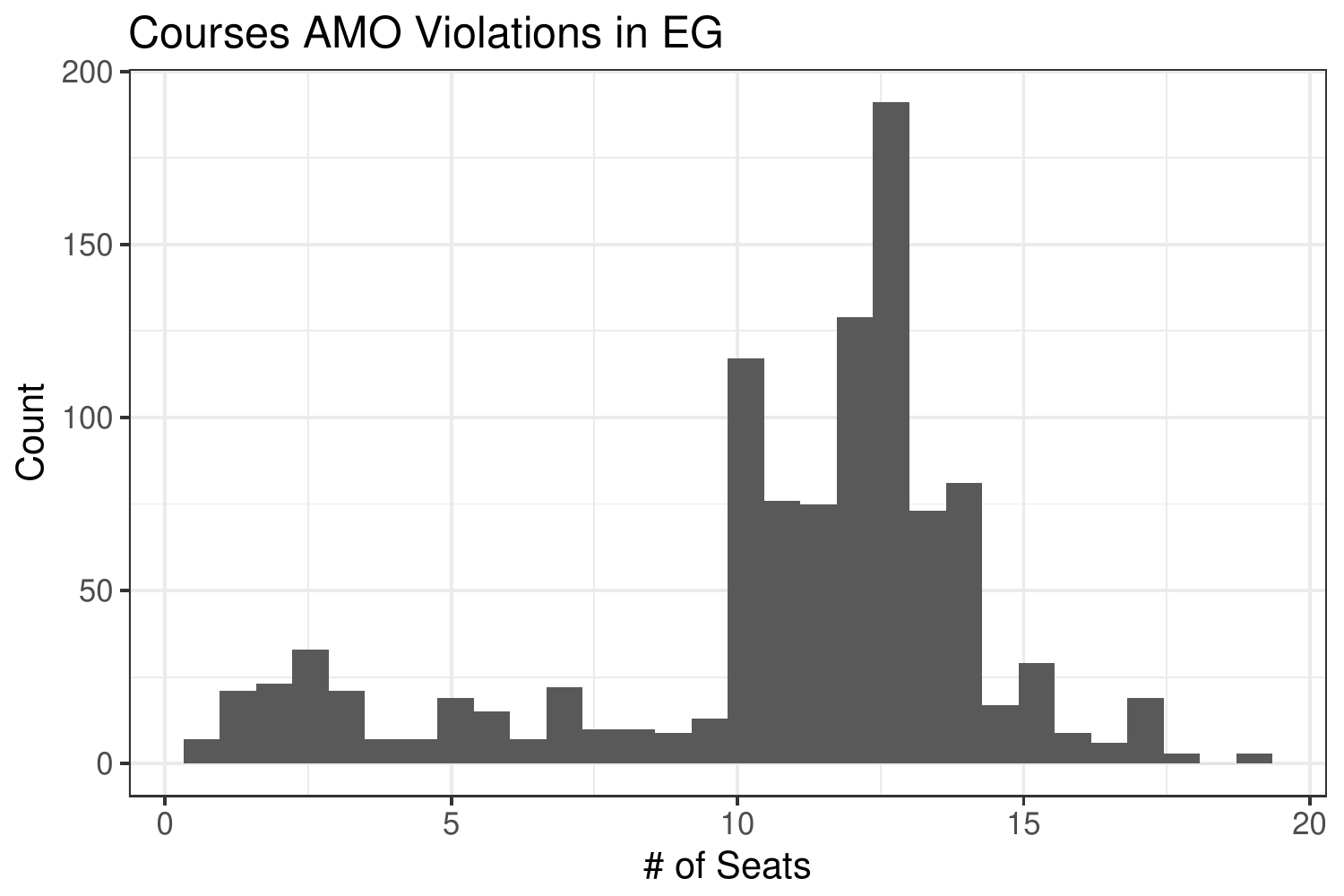}
\end{center}
\caption{Standard EG using linear preferences allocates multiple copies of items to individuals in all of our markets thus is not appropriate for AMO settings. In all markets almost every buyer ends up with an extreme allocation of more than one item.}
\label{eg_violations}
\end{figure}

\subsection{Computational Improvements}
One advantage of the AMO-EG procedure is that it is a result of a convex optimization problem. An approximate version of CEEI (aCEEI) is used in practice to match individuals to courses at leading business schools - the algorithm used for this task is known as CourseMatch (see \url{https://mba-inside.wharton.upenn.edu/course-match/} for a student-facing tutorial of the mechanism). 

The algorithm underlying it is described in \citet{budish2016course}. \cite{budish2011combinatorial} uses the same course data we do and describes the performance of the aCEEI algorithm as: ``the \cite{othman2010finding} computational procedure can solve problems ... 50 courses and schedules...in around 20 minutes in Matlab on a standard workstation. It can solve problems the size of a full year at HBS, in which there are roughly 100 courses and schedules....in around \textit{11 hours} in a more sophisticated computing environment.'' By contrast  the AMO-EG solution on a full year of HBS data can be computed using an industrial grade solver (MOSEK) on a commodity laptop in approximately \textit{15 seconds}.\footnote{Of course, aCEEI allows for more general preferences over bundles than our simple additive ones. However existing benchmarks for aCEEI such as \citet{budish2011combinatorial} already make the assumption that ``(A1) preferences are additively-separable; and (A2) students care about the average rank of the courses they receive'', moreover the same author reports that ``preliminary exploration of preferences more complex than average-rank suggest that all of the reported results are robust'' so actually the linearity assumption is both commonly used and not as serious as it may seem.}

\subsection{Use of Lotteries}
In order to make the allocation space well behaved, AMO-EG uses continuous allocations. Since we use AMO preferences, these can be thought of as lotteries over pure allocations. An important question is whether this convexification leads to very random allocations (many $x_{ij} \in (0,1)$) or whether the resulting allocations are mostly pure. We plot these for all of our markets in Figure \ref{integrality} and see that AMO-EG allocations are mostly pure in practice.

\begin{figure}[h]
\begin{center}
\includegraphics[scale=.42]{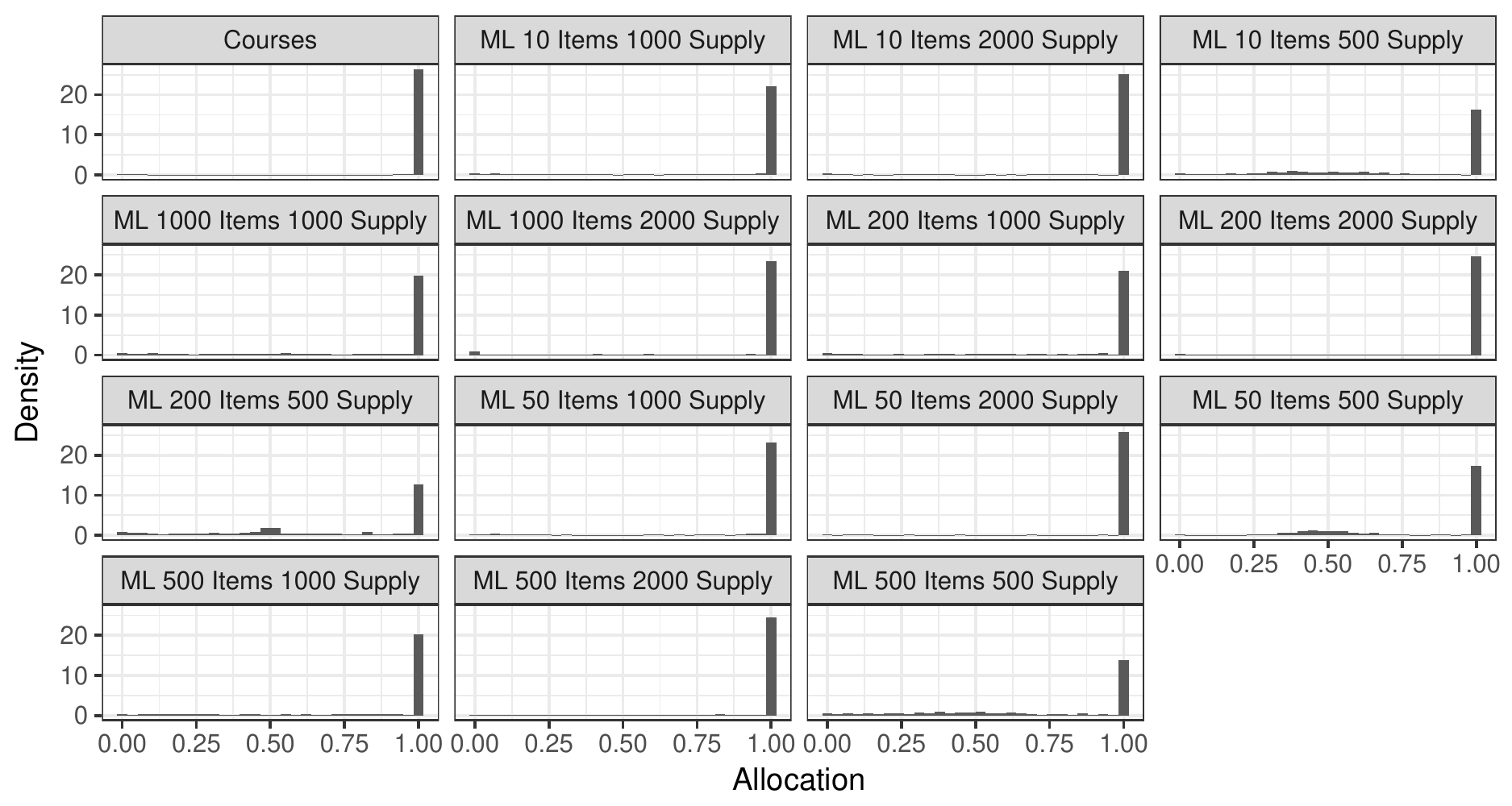}
\end{center}
\caption{Though AMO-EG is allowed to use continuous allocations (lotteries) we see that the resulting allocations are mostly pure across all of our markets. Density plots show density of allocations removing all allocations that are less than $.0001$ (these result from numerical accuracy being on the order of $1e-5$).}
\label{integrality}
\end{figure}

\subsection{Envy and Regret}
We now ask whether the AMO-EG has good properties. We start by considering envy and regret.

First, we define envy for the AMO-EG solution. Let $x^{AMO}$ be the AMO-EG allocation. We compute the envy of individual $i$ as $$\max_{i' \in N} v_i (x^{AMO}_{i'} - x^{AMO}_i).$$ Our measure of envy will be this averaged across all individuals. Equilibria are envy-free (and indeed envy-freeness is an important reason why CEEI is a compelling mechanism) so this is one measure of finite sample performance.

Recall that the EG convex program yielded an equilibrium by using the Lagrange multipliers for the supply constraints at the optimum as the prices that supported its allocation. AMO-EG also has Lagrange multipliers for these constraints (in addition to extra multipliers for each of the AMO constraints). We now ask whether the supply-constraint Lagrange multipliers can be interpreted as equilibrium prices.
We refer to these at the optimum as $p^{AMO}.$ We will take $p^{AMO}$ and for each individual compute their demand $D_i (p^{AMO})$. Recall that the demand may be a set but that all possible allocations in the demand must yield the same utility level. We refer to this as $\bar{u}_i (p^{AMO})$.  We will compute the normalized regret of an individual $i$ facing these prices:

\[ \text{PriceRegret}_i (x^{AMO}, p^{AMO}) = \dfrac{\bar{u} (p^{AMO} - v_i x^{AMO}_i)}{\bar{u}(p^{AMO})}. \]

If an individual $i$'s regret is $0$, it means that $x^{AMO}_i \in D_i (p^{AMO}).$ Note that this is an upper bound on regret, there may exist prices which have lower expected regret for the AMO allocation.

We find that AMO-EG has basically no envy in any of the markets we study (Figure \ref{amo_metrics}). By contrast we see that the Price Regret of AMO-EG goes away as markets grow in the fashion required by our theory. How important price regret is from a practical point of view is an open question since even in the CEEI mechanism the currency and prices are `fake'. Nevertheless, it is one possible measure of how close AMO-EG is to a market equilibrium and perhaps useful in a hybrid mechanism where the mechanism only sets prices and individuals are allowed to buy items as they wish.

\begin{figure}[h]
\begin{center}
\includegraphics[scale=.45]{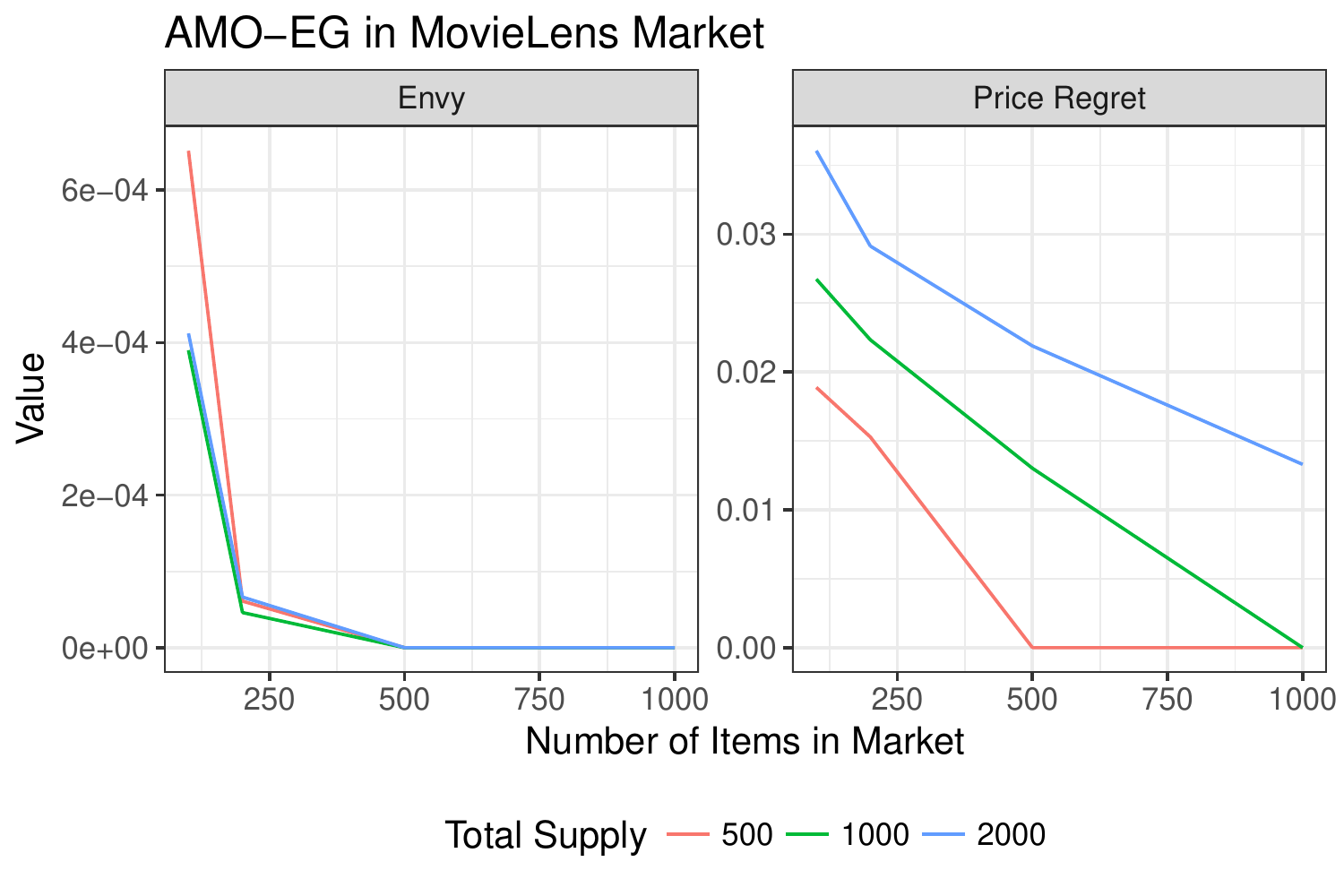}
\includegraphics[scale=.45]{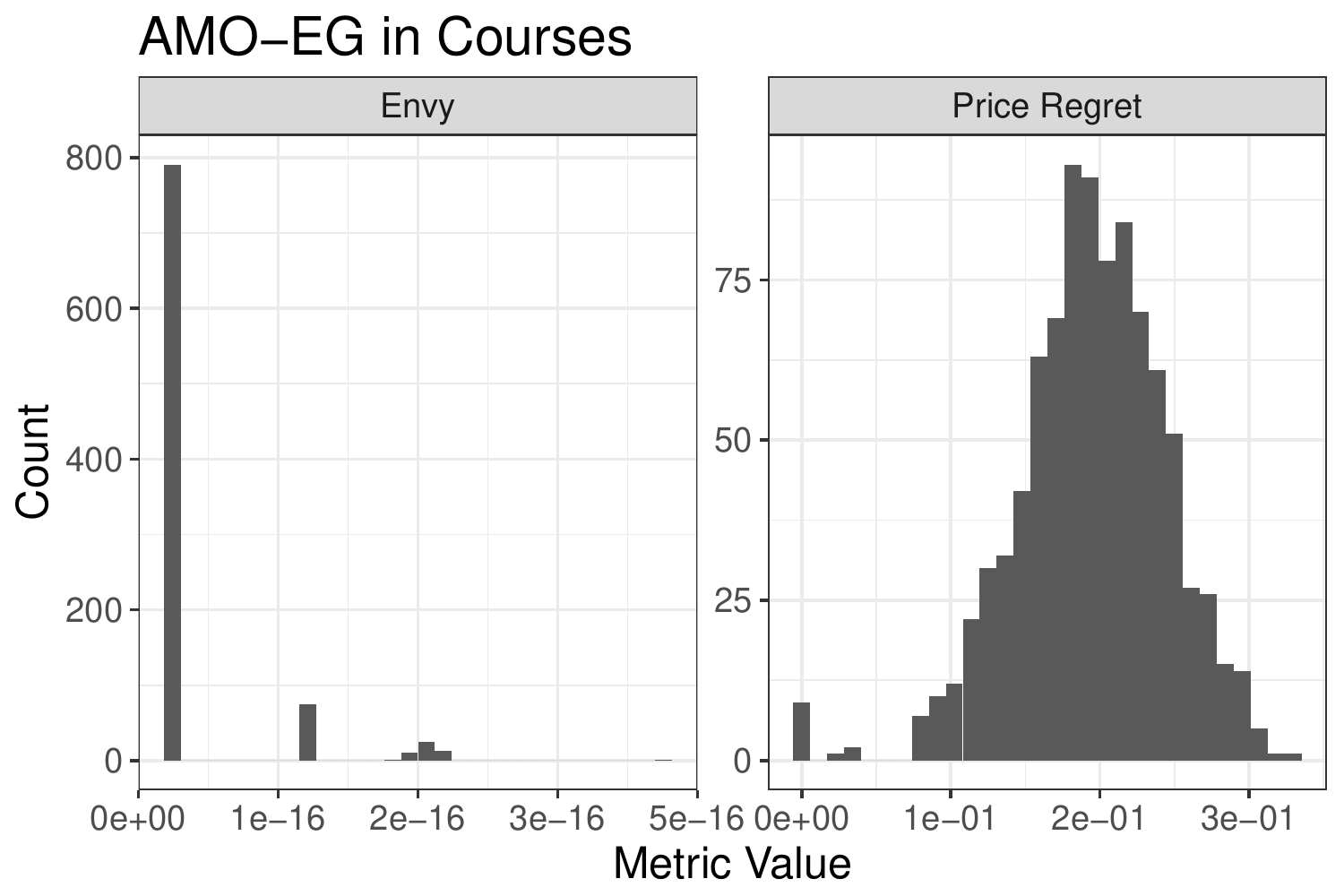}
\end{center}
\caption{AMO-EG yields no envy (up to numerical precision) in all of our markets. When we grow the MovieLens market as our Theorem requires (more items than buyers) we see that Lagrange multipliers from the AMO-EG program can be used as market equilibrium prices. On the other hand in the course allocation problem the prices implied by the Lagrange multipliers are not market clearing prices for the AMO-EG allocation yielding an average of a $20 \%$ regret.}
\label{amo_metrics}
\end{figure}

\subsection{Incentive Compatibility}
\begin{center}
\begin{table*}[h!]\label{ic_results}
\resizebox{\textwidth}{!}{%
 \begin{tabular}{|| c | c | c | c | c ||}
 \hline
 Market & EG LRD Found & AMO-EG LRD Found  & EG LRD Gain & AMO-EG LRD Gain  \\ [0.5ex] 
 \hline\hline
MovieLens (Worst) & .12 (.05)  & .16 (.05) & .007 (.01) & .025 (.02) \\
\hline
MovieLens (Best) & .20 (.05)  & .22 (.06) & .00005 (.0001) & .02 (.02) \\
\hline
Courses & .26 (.06)  & .05 (.03) & .00001 (.0001) & .04 (.03) \\
\hline
\end{tabular}}
\caption{Low rank deviations (LRD) are difficult to find in both AMO-EG and EG. However, we see that even when they are found, LRD lead to only small increases in attained utility for the deviating agents - reported gains are percent increases in utility relative to truthful revelation.} 
\end{table*}
\end{center}
Finally, an important practical reason for the use of CEEI is that it is incentive compatible. However, this property is only guaranteed as the market becomes large. We ask how incentive compatible AMO-EG is in our datasets. In particular, we will ask how much individuals can gain by misreporting their preferences given the fixed valuations of all other individuals. We will compare this to the ability of individuals to deviate in the EG using the same valuation matrix without the AMO constraints.

Ideally, we would like to solve the following problem for each individual. Let $(v'_{i}, v_{-i})$ denote a vector of reports for individual $i$ $(v'_i)$ and for everyone else $(v_{-i})$. Let $x^{AMO} (v'_{i}, v_{-i})$ denote the AMO-EG allocation from these reports. Suppose each individual $i$ has true valuation $v_i$. We would like to measure how much they can gain from reporting non-truthfully. However, computing this quantity is difficult in practice since it requires computing the optimal deviation. 

Since we are not guaranteed convexity of the deviation function, we need to perform a global search. The simplest solution is to discretize the type space and construct a grid-search. However, in high dimensions this quickly becomes intractable. An alternative is to use Bayesian Optimization (BO) techniques which allow for sample efficient exploration of a parameter space \citep{shahriari2015taking}.

Standard BO works as follows: given a space $\Theta$ of possible parameters, we want to maximize a function $f(\theta)$. We sample $K$ points $\theta$ from $\Theta$ and fit a Gaussian Process (GP) model to the observed values $f(\theta).$ GPs give a posterior mean and uncertainty associated with any point $\theta$ in parameter space, so we can use the model uncertainty to search for an optimal next point to sample. We sample this next point, update the model and repeat. A full description of the details involved in each of the steps above is beyond this paper and a large literature addresses many common usecases. For the purposes of this paper we use an open source BO optimization platform deployed at industrial scale (e.g. at facebook) called Adaptive Experiment (Ax, \citep{bakshyae}). Ax is a high level abstraction for BoTorch (\url{https://github.com/pytorch/botorch}) which is itself an extension of the deep learning framework PyTorch \citep{paszke2017automatic}. We use the standard Ax's service API to construct the optimization procedure as detailed in the Ax tutorial here: \url{https://ax.dev/}. We use 50 trials per run with 3 random restarts per person. 

Unfortunately, even with the sample efficient exploration strategy given by BO, our parameter space is intractably large (a valuation for each item). We reduce the parameter space by exploiting the fact that in our data valuations for items are correlated and thus the full valuation matrix can be approximated by a lower rank matrix quite well. The valuation matrix $V$ can be well approximated by the matrix product $V \sim B K^t$ where $B$ is an $N \times d$ matrix of vectors for individuals and $K$ is an $M \times d$ matrix of vectors for items so $v_{ij}$ can be well approximated by $b_i k_j.$  We consider deviations by individuals only across these principal directions of variation that are fixed by the matrix $K.$

We allow deviations where individuals choose $b_i$ from a compact subset of $\mathbb{R}^d$ which then gives us a full deviation vector of the form $v'_{ij} = v_{ij} + b_i \cdot k_j.$ Here choosing $0^d$ corresponds to truth telling. We refer to these as low-rank deviations.

If $d < M$ then we have reduced the size of the parameter space that we need to search. In practice we will use $d=10.$ We justify the focus on low rank deviations by the fact that in many large scale systems (e.g. recommender systems, online advertising) it is infeasible to report a full vector of valuations. In practice such systems have missing valuations which are filled in by assuming a low rank structure and using matrix completion. Since an individual $i$ is small relative to the market changing their report won't change the found latent space so indeed they can only deviate in the low rank space. 

Even with these heavy restrictions, the optimization process is quite computationally expensive as it requires solving for the AMO-EG/EG equilibrium at each time step. We consider three instances: the `worst case' MovieLens marker (200 buyers, 50 items, total supply 2000), the `best case' MovieLens market (200 buyers, 1000 items, total supply 2000) and the HBS course market. 

In each replicate we sample an individual from the market, hold the others fixed, and use BO as described above to try to find a low rank deviation which increases that individual's utility. We repeat this for $49$ individuals from each market. We see that finding low rank deviations (LRDs) is quite difficult (Table \ref{ic_results}) - only about $5$ to $20 \%$ of the time does BO find any deviation from truthtelling and even when it does the gains are no more than $5$ percent of the baseline utility.

This analysis provides no guarantee that bigger deviations do not exist. Indeed, perhaps profitable deviations exist in the cut off parts of the spectrum or perhaps BO is not a good optimizing technique to find them. We also tried using standard zeroth order methods provided with SciPy \citep{jones2001scipy}. We found these to perform even worse both with initialization at truthtelling and with random restarts. In addition, we tried looking for low rank deviations in a different mechanism: one which inputs valuations and outputs the social efficiency (sum of utilities) maximizing allocation. Here there is an obvious profitable deviation - report high values for all items and receive all of them. We found that in this setting BO easily found optimal low rank deviations. Developing better methods for computing the incentive compatibility of a particular instance is an important topic for future work.
\section{Conclusion}
The multi-unit assignment problem is heavily studied both in theory and in practice. We have shown that a convex program can be used to efficiently apply the CEEI mechanism in the case of linear, at-most-one preferences with divisible goods (or lotteries). This is a natural preference class for many real world allocation problems and our results allow CEEI to scale to very large assignment tasks.

Many open questions remain: can we incorporate substitutes and complements into the CEEI mechanism in a tractable way? What other preference classes allow for simple scalable solutions to the CEEI problem? How can we take advantage of the special structure of the EG convex program to solve it more efficiently for large low-rank markets?

\bibliographystyle{ACM-Reference-Format}
\bibliography{main_ceg.bbl}

\appendix

\section{Eisenberg-Gale for Constraint Groups}
We will describe our results for a more general model than the AMO model used in the main body of the paper. Instead we will assume that we have a set $\mathcal K$ of \textit{constraint groups} which is a partitioning of items. Thus each $k \in \mathcal K$ is a subset of the original items. We assume that each buyer can be assigned at most 1 unit of items among those in the group $k$. This nests AMO preferences by letting $k$ be singletons. Furthermore, we allow general non-equal budgets rather than assuming that all budgets are equal as in the main text (equal budgets is equivalent to assuming that all budgets are $1$).

The model we wish to solve is the following
\begin{equation}
\begin{array}{rrclcl}
\max_{x} & \multicolumn{3}{l}{\sum_{i} B_i \log(v_i \cdot x_i)} \\
\textrm{s.t.} & \sum_i x_i & \leq & s \\
&\sum_{j \in k} x_{ij} & \leq & 1 & & \forall i \in N, k \in \mathcal K \\
& x & \geq & 0 & &  \\
\end{array}
\end{equation}

Writing the Lagrangian we get 
\begin{align}
  \label{eq:eg_minmax_multislot}
  \min_{p, \lambda \geq 0} \max_{x \geq 0} \sum_i B_i \log(v_i\cdot x_i) - \sum_j p_j\left(\sum_i x_i - 1 \right) - \sum_{i,k} \lambda_{ik} \left(\sum_{j \in k} x_{ij} - 1\right)
\end{align}

Let $\beta_i = \frac{B_i}{v_i\cdot x_i}$ be the \emph{utility price} of buyer $i$. The first-order optimality condition wrt. $x_{ij}$ gives
\begin{align}
  \label{eq:price_auction_bound}
  v_{ij} \beta_i - p_j - \lambda_{ik} \leq 0
  \Rightarrow 
  v_{ij} \beta_i \leq p_j + \lambda_{ik}  
\end{align}
Furthermore, if $x_{ij} > 0$ we have
\begin{align}
  \label{eq:eg_xij_fo}
  v_{ij} \beta_i = p_j + \lambda_{ik},
\end{align}
which shows that if $x_{ij}>0$ then the utility price $\frac{p_j + \lambda_{ik}}{v_{ij}}$ is tied for best among all items.

Multiplying each equation in \eqref{eq:eg_xij_fo} by $x_{ij}$ and summing over $j$ gives
\begin{align}
  \label{eq:budget spent}
   \sum_j x_{ij}(p_j + \lambda_{ik})
   = \sum_{j} x_{ij} v_{ij} \frac{B_i}{v_i\cdot x_i}
   = B_i
\end{align}
If not for $\lambda_{ik}$, this would show that each buyer is spending their entire budget under prices $p$. Coupled with \eqref{eq:eg_xij_fo} this shows that the entire budget is spent on optimal items, and thus each buyer is receiving a bundle from their demand. Thus when $\lambda = 0$ we get a market equilibrium. However, when $\lambda_{ik}>0$ for some $i,k$ we no longer get a market equilibrium under prices $p$ because the buyer is spending less than their budget. If we allow additional personalized prices for each group then we could set those equal to $\lambda_{ik}$ and we would get a form of ``market equilibrium with personalized prices'' under $p,\lambda$. But such a market equilibrium concept with personalized prices does not seem particularly compelling.
Instead, we will show that for an appropriate, and practical, notion of ``large'' markets each $\lambda_{ik}$ becomes small as the market size increases. Thus we will show that our variant of the Eisenberg-Gale convex program gives an approximate market equilibrium for large markets, where the approximation error is negligible for sufficiently-large markets.

Our result will be based on a structure satisfied by large markets: each buyer has an increasingly-large set of approximate ``twins,'' meaning buyers with approximately the same valuation. Our theorem will show that for each buyer $i$, as long as some approximate twin is unconstrained for constraint group $k$, then $\lambda_{ik}$ is approximately zero.
The following is our main theorem

\begin{theorem}[Approximate twins give bounded AMO prices]
  \label{thm:approx twin approx ceei}
  Let $(x,p,\lambda)$ be a solution to \eqref{eq:eg_minmax_multislot}. For any buyer $i$ and constraint group $k$ such that $\lambda_{ik}>0$, let $i'$ be another buyer such that $| \beta_i - \beta_{i'}| \leq \epsilon_{\beta}$,  and $\sum_{j \in k} x_{i'j} < 1$. Let $j_k\in k$ be such that $x_{ij_k}>0$ and $| v_{ij_k} - v_{i'j_k} | \leq \epsilon_v^k$.
  Then we can bound the value of $\lambda_{ik}$ as follows
  \[
    \lambda_{ik} \leq 
    \epsilon_{\beta} v_{ij_k} + \beta_{i'}\epsilon_v^k
  \]
\end{theorem}
\begin{proof}
  Applying \eqref{eq:price_auction_bound} for $i'$ and noting that $\lambda_{i'k}$, along with the bounds $\epsilon_v^k, \epsilon_{\beta}$, gives
  \begin{align*}
    p_j \geq \beta_{i'} v_{i'j_k}
    \geq \beta_{i'} (v_{ij_k} - \epsilon_v^k)
    \geq (\beta_{i} - \epsilon_{\beta}) v_{ij_k} - \beta_{i'}\epsilon_v^k
    = \beta_{i}v_{ij_k} - \epsilon_{\beta} v_{ij_k} - \beta_{i'}\epsilon_v^k
  \end{align*}

  Now we can apply \eqref{eq:eg_xij_fo} for buyer $i$ and apply our bound on $p_{j_k}$ to show the theorem
  \begin{align*}
    &v_{ij_k} \beta_i = p_{j_k} + \lambda_{ik}
    \geq  \beta_{i}v_{ij_k} - \epsilon_{\beta} v_{ij_k} - \beta_{i'}\epsilon_v^k
    + \lambda_{ik} \\
    \Leftrightarrow
    & \lambda_{ik}
    \leq   \epsilon_{\beta} v_{ij_k} + \beta_{i'}\epsilon_v^k
  \end{align*}
\end{proof}

If every buyer, constraint group pair $i,k$ has an associated unconstrained approximate twin then Theorem~\ref{thm:approx twin approx ceei} allows us to bound the error in several of the properties that are usually guaranteed by EG, as well as by CEEI and market equilibrium in general.

\subsection{Helper lemmas}
In order to show when Theorem~\ref{thm:approx twin approx ceei} gives useful bounds we will need upper and lower bounds on several basic quantities suchs as prices, utilities, and  utility-prices. For most of the lemmas in this section, the point is that we need some upper and lower bound on the quantity in question, such that the bound \emph{does not increase with market size}. Thus we can in turn use these lemmas in later sections where ratios such as $\frac{\max_{ij}v_{ij}}{\min_{ij}v_{ij}}$ or $\frac{\sum_i B_i}{\sum_j s_j}$ are bounded by a constant that does not grow with market size.
We will use the helpful assumption that $v^{\downarrow}\defeq \min_{ij} v_{ij} > 0$.

First we show that each utility price $\beta_i$ is bounded both above and below by market constants.
\begin{lemma}
  \label{lem:beta ub}
  For any buyer $i$ we have $\frac{\sum_{i'}B_{i'}}{(v^{\uparrow})^2 \sum_{j} s_j}v^{\downarrow} \leq \beta_i \leq \frac{\sum_{i'}B_{i'}}{|\mathcal{K}|v^{\downarrow}}$
  
\end{lemma}
\begin{proof}
  For the lower bound:
  We know that the sum of payments $\sum_{ij} \beta_iv_{ij}x_{ij}$ in equilibrium is exactly $\sum_{i'}B_{i'}$ so we have
  \[
    \sum_{i'}B_{i'} = \sum_{i',j} x_{i'j} \beta_{i'} v_{i'j}
    \leq
    \max_{i'}\beta_{i'} v^{\uparrow} \sum_{i',j} x_{i'j} 
    =
    \max_{i'}\beta_{i'} v^{\uparrow} \sum_{j} s_j
  \]
  which implies $\max_{i'}\beta_{i'} \geq \frac{\sum_{i'}B_{i'}}{v^{\uparrow} \sum_{j} s_j}$. This shows that the price on every item is lower bounded by $\frac{\sum_{i'}B_{i'}}{v^{\uparrow} \sum_{j} s_j}v^{\downarrow}$; thus for $i$ to win \emph{any} items, they must have $\beta_iv_{ij} \geq \frac{\sum_{i'}B_{i'}}{v^{\uparrow} \sum_{j} s_j}v^{\downarrow}$ for some item $j$. The lower bound follows.

  For the upper bound:
  We know that the sum of prices must be lower than $\sum_{i'}B_{i'}$.  We can thus deduce that $\sum_k \beta_i \max_{j \in k} v_{ij} \leq \sum_{i'}B_{i'}$, otherwise there exists a single item from each group such that the sum of payments on just a single unit of each of those items sums to more than the total amount of budget in the market. It follows that $\beta_i \leq \frac{\sum_{i'}B_{i'}}{|\mathcal{K}|v^{\downarrow}}$
\end{proof}

From the bounds in Lemma~\ref{lem:beta ub} we can deduce the following upper and lower bounds on utility for all $i$ 
\begin{align}
  \label{eq:u ub}
  u_i = B_i\beta_i^{-1} \leq B_i \frac{\sum_{i'}B_{i'}}{(v^{\uparrow})^2 \sum_{j} s_j}v^{\downarrow} \eqdef u_i^{\uparrow} \\
  \label{eq:u lb}
  u_i = B_i\beta_i^{-1} \geq B_i \frac{|\mathcal{K}|v^{\downarrow}}{\sum_{i'}B_{i'}}v^{\downarrow} \eqdef u_i^{\downarrow}
\end{align}

We can use this lemma to also give a bound on the total amount of supply that any given buyer can be allocated:
\begin{lemma}
  \label{lem:bounded total allocation}
  For any buyer $i$ we have $\sum_{j} x_{ij} \leq \frac{u^{\uparrow}_i}{v^{\downarrow}}$
\end{lemma}
\begin{proof}
  This follows from the upper bound on utility in \eqref{eq:u ub}
  \[
    u_i^{\uparrow} \geq u_i
    =
    \sum_{j} x_{ij} v_{ij}
    \geq
    \sum_{j} x_{ij} v^{\downarrow}
  \]
\end{proof}

Now we will show that for every buyer-pair we can bound their difference in utility prices based on how similar their valuation vector is. This lemma combines with Theorem~\ref{thm:approx twin approx ceei} to show that if every buyer $i$ has a large-enough set of approximate twins (such that some approximate twin is not constrained for constraint group $k$), then they have a set of competitors that will force $\lambda_{ik}$ to zero for each $k$.
\begin{lemma}
  \label{lem:beta bounds}
  Let $p$ be a set of prices such that $\sum_j s_jp_j \leq n$,
  and let $\beta$ be a vector of utility prices such that every buyer $i$ has an associated $x_i$ such that $\sum_j x_{ij}(p_j+\lambda_{ik})\leq B_i$ and \eqref{eq:eg_xij_fo} is respected for all $x_{ij}>0$. Then we have the following multiplicative and additive bounds on utility-price differences between buyer pairs $i,i'$ such that $\| v_i - v_{i'} \|_{\infty} \leq \epsilon$
  \begin{align}
    \label{eq:beta mult bound}
     \epsilon^{\downarrow} \beta_{i'} \leq \beta_{i} \leq \epsilon^{\uparrow} \beta_{i'},
  \end{align}
  \begin{align}
    \label{eq:beta add bound}
    |\beta_i - \beta_{i'}| \leq (\epsilon^{\uparrow} - 1) \frac{\sum_{i'}B_{i'}}{|\mathcal{K}v^{\downarrow}},
  \end{align}
where $\epsilon^{\downarrow} = \min_{i^*\in\mathcal{I},j \in J} \frac{v_{i^*j'}}{v_{i^*j'} + \epsilon}$ and 
  $\epsilon^{\uparrow} = \max_{i^*\in\mathcal{I},j \in J} \frac{v_{i^*j'}}{v_{i^*j'} - \epsilon}$, with $\mathcal{I} = \{i,i'\}$.
\end{lemma}
\begin{proof}
  We wish to bound $\beta_i$ in terms of $\beta_{i'}$. First, if $\beta_iv_{ij} > \beta_{i'}v_{i'j}$ for all items $j$ that $i'$ obtain under $\beta'$ then $i$ must be spending strictly more than their budget, since they would win all items that $i'$ wins, and pay strictly more. Thus, there must exists some item $j$ such that $\beta_iv_{ij} \leq \beta_{i'}v_{i'j}$, which implies $\beta_i \leq \beta_{i'}\frac{v_{i'j}}{v_{i'j}-\epsilon}$. The converse argument shows that there must exist some  $j'$ such that $\beta_{i} \geq \beta_{i'}\frac{v_{ij'}}{v_{ij'} + \epsilon}$.  We thus have the multiplicative bound
  \begin{align}
    \label{eq:beta sandwich}
    \frac{v_{ij'}}{v_{ij'}+\epsilon} \beta_{i'} \leq \beta_{i} \leq \frac{v_{i'j}}{v_{i'j} - \epsilon} \beta_{i'}.
  \end{align}
  
  The additive bound follows from the multiplicative bound because (choosing $i,i'$ arbitrarily)
  \[
    \beta_{i'} - \beta_i \leq \epsilon^{\uparrow} \beta_i - \beta_i
    \leq (\epsilon^{\uparrow}-1) \beta_i \leq (\epsilon^{\uparrow}-1) \frac{\sum_{i'}B_{i'}}{|\mathcal{K}|v^{\downarrow}}.
  \]
\end{proof}
Note that Lemma~\ref{lem:beta bounds} does not require supply constraints to be respected, it only requires that buyers pay according to $\beta_iv_{ij}$ and buy items when they bid more than the price $p$. Thus the lemma applies to market equilibria, but also more broadly to allocations attained via picking the budget-exhausting $\beta_i$ given prices $p$.

\subsection{Regret bounds}

Now we can show that the bounds on AMO prices from Theorem~\ref{thm:approx twin approx ceei} lead to an approximate market equilibrium, in the sense that the regret of each buyer is bounded.

\begin{corollary}
  \label{cor:regret bound}
  Let $(x,p,\lambda)$ be the solution to \eqref{eq:eg_minmax_multislot}, and let $x,p$ be used as an approximate market equilibrium. For any buyer $i$ such that for every constraint group $k$ they have an $\epsilon$-twin $i_k$ such that $\sum_{j\in k}x_{i_kj}<1$, regret is bounded by $\frac{u_i^{\uparrow}}{v^{\downarrow}}\max_k \beta_i^{-1} \left(  \eps_{\beta}v_{ik} + \beta_{i_k}\epsilon_v^k\right)$.
  When $\frac{s_j}{n} \leq 1$ for all $j$ then the proportional share gap is bounded by the same quantity.
  For the equal-budgets setting envy is bounded by the same quantity.
\end{corollary}
\begin{proof}
  We first show the result for regret. For any constraint group $k$ such that $\lambda_{ik}>0$ we know from \eqref{eq:price_auction_bound} that $i$ already bought a whole unit of their best bang-per-buck item from group $k$. Thus, at best $i$ can spend their leftover budget $\sum_j x_{ij}\lambda_{ik}$ on items with utility price $\beta_i$ or worse. Since every item they currently buy has utility price $\beta_i$ or better, their regret is bounded by
  \begin{align*}
    \beta_i^{-1} \sum_{j} x_{ij} \lambda_{ik_j}
    \leq
    \beta_i^{-1}  \sum_{k} \sum_{j\in k} x_{ij}  (\eps_{\beta}v_{ik} + \eps_v^k \beta_{i_k})
    \leq
     \frac{u^{\uparrow}}{v^{\downarrow}}\max_{k} \beta_i^{-1} \left( \eps_{\beta}v_{ik} + \beta_{i_k}\eps_v^k \right).
  \end{align*}
  The first inequality follows from Theorem~\ref{thm:approx twin approx ceei}. The second inequality follows by taking the max over terms in parentheses and using Lemma~\ref{lem:bounded total allocation} to bound $\sum_{k, j\in k}x_{ij}$.

  To show the proportional share result, note that a demand solution for buyer $i$ receives at least the proportional share; since $\sum_j s_jp_j \leq \sum_{i'}B_{i'}$ and $\frac{s_j}{n}\leq 1$ buyer $i$ can buy $\frac{s_j}{n}$ of each item and be guaranteed to satisfy both the AMO and budget constraints. But weknow from the bound on regret that buyer $i$ receives a bundle that is within the stated distance to a demand solution, and thus the bound on the proportional share gap follows.

  To show the result for envy, simply note that a demand solution has zero envy when budgets are equal. Since our current solution satisfies the stated bound to a demand  solution, the envy must be bounded by the same quantity.
\end{proof}

The above bound on envy uses the observation that in market equilibrium buyers get an allocation from their demand set, and for equal budgets this means that they must weakly prefer their own allocation. In our experiments that we show later, we find that there is no envy, even when regret is non-negligible. That may be due to the observation that envy can be bounded by the weakly stronger statement
\[
  \text{Envy}_i \leq \sum_k \beta_i^{-1} \left(  \eps_{\beta}v_{ik} + \beta_{i_k}\epsilon_v^k\right) - \min_{i'} v_i \cdot (x^* - x_{i'})
\]
where $x^*$ is a demand allocation for buyer $i$. This statement shows that if $i$ strictly prefers their demand allocation to the current allocation of others, then that non-zero difference cancels out some of the error due to not getting an allocation from their demand set.

Corollary~\ref{cor:regret bound} has a dependence on the utility prices $\beta_i, \beta_{i_k}$, but we know from Lemma~\ref{lem:beta bounds} that this dependence can be removed. These facts imply Theorem~\ref{maintheorem}.

\subsubsection{Proof of Theorem~\ref{maintheorem}}

We will show a slightly more general version of Theorem~\ref{maintheorem} which holds for constraints groups rather just the AMO setting. Let $s_{max} = \max_{k \in \mathcal{K}} \sum_{j \in k} s_j$.

\begin{proof}
  Let $\delta \geq 0$ be some desired approximate equilibrium. Since every buyer has at least $s_{max}+1$ $\epsilon$-twins we have that for each buyer and constraint group $k$ there is at least one $\epsilon$-twin for whom constraint $k$ is not binding.
  By Corollary~\ref{cor:regret bound} we then have that the regret of buyer $i$ is bounded by
  \[
    \frac{u_i^{\uparrow}}{v^{\downarrow}}\max_k \beta_i^{-1} \left(  \eps_{\beta}v_{ik} + \beta_{i'}\epsilon_v^k\right)
    \leq
    \frac{u_i^{\uparrow}}{v^{\downarrow}}\max_k \frac{(v^{\uparrow})^2\sum_j s_j}{\sum_{i'}B_{i'}} \left(  \eps_{\beta}v_{ik} + \frac{\sum_{i'}B_{i'}}{|\mathcal{K}|v^{\downarrow}}\epsilon\right)
  \]
  where the inequality follows by Lemma~\ref{lem:beta ub}. Now we can apply Lemma~\ref{lem:beta bounds} to bound this by
  \[
    \frac{u_i^{\uparrow}}{v^{\downarrow}}\max_k \frac{(v^{\uparrow})^2\sum_j s_j}{\sum_{i'}B_{i'}} \left( (\epsilon^{\uparrow} - 1) \frac{\sum_{i'}B_{i'}}{|\mathcal{K}v^{\downarrow}} v_{ik} + \frac{\sum_{i'}B_{i'}}{|\mathcal{K}|v^{\downarrow}}\epsilon\right)
  \]
where $\epsilon^{\uparrow} = \max_{i'\in\mathcal{I},j \in J} \frac{v_{i'j'}}{v_{i'j'} - \epsilon}$, where $\mathcal{I}$ is the set of $\epsilon$-twins for $i$.
  For a sufficiently small $\epsilon$ this expression will be less than $\delta$.
\end{proof}

\subsection{Large-market results}

We will investigate the low-rank market setting from Section~\ref{sec:amo}. Armed with Corollary~\ref{cor:regret bound} we know that if every buyer has greater than $\max_j s_j + 1=s+1$ approximate twins then their regret is bounded by a market-size independent constant that only depends on how similar their approximate twins are. Thus, we may take a discretization of the low-rank spaces $\Theta,\Psi$ into a finite set of bins (this can also be thought of as a gridding, with each bin corresponding to a hypercube in the gridding).
For each bin, we will consider any pair of buyers whose types $\theta_i,\theta_{i'}$ belong to the same bin to be approximate twins. Let $\Theta'$ be the set of types corresponding to some bin. Since we know the maximum possible distance between types in $\Theta'$ we get an upper bound on valuation differences of the form $\|v_i - v_{i'}\|_{\infty} \leq \max_{\theta, \theta' \in \Theta'} \max_{\psi \in \Psi} \|\psi \cdot (\theta - \theta') \|_\infty$. Since the low-rank spaces are compact this is bounded by a constant that depends on the bin size.
We may make our discretization arbitrarily fine, and as the discretization gets finer we get tighter and tighter bounds on $\epsilon_{\beta}$ and $\epsilon^k_v$ for pairs of buyers that share a bin. For a sufficiently fine discretization we can then sample larger and larger markets; for some sufficiently large market each bin has at least $s+1$ buyers in it with exponentially high likelihood, which we show below for completeness.

We can think of such a binning as a multinomial distribution where each variable in the multinomial distribution corresponds to a bin in the gridding.
\begin{lemma}
  \label{lem:multinomial}
  Let $x_1,\ldots,x_K$ be random variables denoting the number of balls in each of $K$ bins, with balls distributed according to a multinomial distribution with bin-probabilities $\gamma_1,\ldots, \gamma_K$. The probability that bin $i$ has less than $\frac{\mathbb{E}x_i}{c}$ balls can be bounded as
  \begin{align*}
    \bbP\left(x_i \leq \frac{\mathbb{E}x_i}{c}\right) \leq \exp\left(-2n\gamma_i^2 \left(\frac{c-1}{c}\right)^2\right)
  \end{align*}
The probability that any bin has less than $\frac{\mathbb{E}x_i}{c}$ balls can be bounded with a union bound as
  \begin{align*}
    \bbP\left(\exists i: x_i \leq \frac{\mathbb{E}x_i}{c} \right)
    \leq K \exp\left(-2n\min_i\gamma_i^2 \left(\frac{c-1}{c}\right)^2\right) 
    \leq \exp\left(\log(K)-2n\min_i\gamma_i^2 \left(\frac{c-1}{c}\right)^2\right) 
  \end{align*}
\end{lemma}
\begin{proof}
  First we note that the expected number of balls in bin $i$, $\bbE(x_i)$, can be viewed as a sum of indicator random variables, and thus $\bbE(x_i) = n\gamma_i$. 
  Thus having only $\frac{\mathbb{E}x_i}{c}$ ball in bin $i$ means that $x_i$ deviates from its mean by $\mathbb{E}x_i (1 - \frac{1}{c})$. Now we can apply Hoeffding's inequality to obtain the first result:
  \begin{align*}
    \bbP\left(x_i \leq \frac{\mathbb{E}x_i}{c}\right) \leq \exp(-2(n\gamma_i(1-\frac{1}{c}))^2 / n)
  \end{align*}
  To obtain the second result we can simply use a union bound based on the first result.
\end{proof}


Thus, in low-rank markets every buyer has a large set of approximate twins, and thus the conditions in Lemma~\ref{lem:beta bounds} are satisfied for a large group of buyers. Simultaneously, in this market model the supply of each item remains bounded by a constant, and so for any buyer $i$ and constraint group $k$, there must exist an approximate twin $i'$ such that they have $\sum_jx_{ij}<1$.  It now follows from Corollary~\ref{cor:regret bound} that in large low-rank markets solutions to AMO-EG are approximate market equilibria in the sense that the regret of each buyer goes to zero almost surely as the market gets large.

\subsubsection{Proof of Theorem~~\ref{maintheorem probabilistic}}
\begin{proof}
  By Theorem~\ref{maintheorem} there exists an $\epsilon$ such that when every buyer has $s+1$ $\epsilon$-twins then $ \left(x_{AMO},p_{AMO}\right) $ is a $\delta$-approximate CEEI. Now consider a discretization of the buyer type space into regions such that for each region any pair of types within that region are guaranteed to be $\epsilon$-twins. Such a discretization is guaranteed to exist due to the type space being of fixed finite dimension $d$ and compact. But then we can apply Lemma~\ref{lem:multinomial} to get that the probability of having each region contain less than half its expected number of buyers decreases exponentially in market size. Thus for a sufficiently large market each buyer has a constant fraction of $n$ $\epsilon$-approximate twins with exponentially-high likelihood in market size. In particular, we may pick $n$ such that the probability is at least $\delta$.
\end{proof}

\section{IC in the large results}

In this section we will show that the mechanism resulting from using the solutions to EG and AMO-EG for allocation are both \emph{strategy-proof in the large} (SP-L) for the low-rank setting. This notion of strategyproofness was introduced by \citet{azevedo2018strategy} as a way to show that some mechanisms become approximately incentive compatible (IC) as the market gets large, even if they are not IC in general. Intuitively speaking, the requirement is that for any strategy profile (i.e. a mapping from type to reported type) such that the resulting reported-type distribution has full support and $\epsilon$, there exists a sufficiently large market such that every buyer can gain at most $\epsilon$ by misreporting.

Our SPL proof will consist of four steps. First, we show a certain monotonicity condition on the $\beta$ variables in AMO-EG. Secondly, we show that each price becomes lower bounded by a market-size-independent constant with high probability. Thirdly, using the monotonicity result we show that each buyer can have only a negligible effect on price with high probability. Finally, we use this to conclude that each buyer cannot change their utility by much.

We now show a monotonicity condition on $\beta$ in the at-most-one setting. In particular, we will interpret the utility prices $\beta$ as pacing multipliers in an allocation model that uses auctions. We will show that taking the element-wise maximum over two budget-feasible $\beta, \beta'$ retains feasibility in this setting. Similar results were previously shown for the standard EG model~\citep{peysakhovich2019fair}, and originally shown for a quasi-linear model that explicitly uses auctions~\citep{conitzer2019pacing}.

\begin{defn}[Auction allocation via paced bids]
  For a given $\beta$, \emph{paced auction allocation} sets each price $p_j$ to the $s_j$'th highest bid, and sets $x_{ij}=1$ if $\beta_iv_{ij}$ is among the $s_j$ highest bids (for $i$ tied for $s_j$'th we allow fractional allocation). Set $\lambda_{ij} = \beta_iv_{ij} - p_j$. Utility prices $\beta$ are \emph{budget-feasible pacing multipliers} (BFPM) if the corresponding $x,p,\lambda$ leads to an allocation such that $\sum_j x_{ij}(p_j+\lambda_{ij}) \leq B_i$ for all buyers $i$.
\end{defn}
\begin{proposition}
  \label{prop:beta monotone}
  The utility prices $\beta$ satisfy the following monotonicity condition in the AMO-EG setting. If $\beta,\beta'$ are both budget-feasible under paced auction allocation then $\beta^+_i = \max(\beta_i \beta_i')$ is budget-feasible.
\end{proposition}
\begin{proof}
  Assume WLOG that $\beta_i$ attains the maximum for $i$ under $\beta^+$. Then $i$ pays the same price $\beta_iv_{ij}$ for any item they win under $\beta^+$. Now let $x^+_{ij}=x_{ij}$ if $i$ is still among the top $s_j$ bids. If $i$ is no longer among the top $s_j$ bids then $x^+_{ij}=0$. Note that $i$ cannot be among the top $s_j$ bids for any item $j$ under $\beta^+$ if they were not among the top $s_j$ under $\beta$, since the bids of $i$ remain the same, and every other bid weakly increased. Thus, $i$ gets an allocation $x^+_i \leq x_i$, and since their personalized price $p_j+\lambda_{ij}$ stays the same, $i$ must remain budget feasible.
\end{proof}

Proposition~\ref{prop:beta monotone} also shows that the maximal BFPM is unique, since they come from a compact space (by Lemma~\ref{lem:beta ub}).

Second, we show that a vector of utility prices is the AMO-EG solution $\beta$ if and only if it is the pointwise maximal BFPM (there is a single pointwise maximal BFPM by compactness of utility prices).
\begin{proposition}
  \label{prop:eg beta maximal}
  A vector of utility prices $\beta$ is a solution to AMO-EG if and only if $\beta$ equals the pointwise max of $\beta$ and $\beta'$ for all BFPM $\beta'$.
\end{proposition}
\begin{proof}
  First, let $\beta$ be the utility prices associated with a solution $(x,p)$ to AMO-EG.
  We may assume that $p_j$ equals the least bid $\beta_iv_{ij}$ such that $x_{ij}>0$, since otherwise we may increase $p_j$ until this is true, while lowering all $\lambda_{i'j}$ by the corresponding amount; this retains optimality.
  By \ref{eq:price_auction_bound} and \ref{eq:eg_xij_fo} we have that only ``winning'' bids have $x_{ij}>0$.
  Now assume that $\beta$ were \emph{not} maximal, then there exists $\beta'$ s.t. $\beta'_i>\beta_i$ for some $i$. Taking the pointwise max of $\beta'$ and $\beta$ would then yield a new BFPM, but then spending went strictly up for at least one buyer, which is a contradiction since \ref{eq:budget spent} guarantees that each buyer spends their budget exactly at $\beta$. This shows the only if direction.

  We have that a solution $\beta$ associated with AMO-EG is maximal. But by uniqueness of the maximal BFPM and the fact that AMO-EG is guaranteed to have an optimal solution (since it is guaranteed to be feasible), we have that $\beta$ must also be the only maximal BFPM, thus showing the if direction.
\end{proof}

\subsubsection*{Prof of Theorem~\ref{thm:spl}}

We prove Theorem~\ref{thm:spl} using the monotonicity of utility prices proved in Propositions~\ref{prop:beta monotone} and~\ref{prop:eg beta maximal}. In particular, we will show that in the case where a given buyer has sufficiently-many approximate twins, the utility-price monotonicity implies that they cannot change prices by much, and cannot gain much by tailoring their valuations to specific prices. The case where they have a small number of approximate twins goes to zero rapidly.

\begin{proof}
  Let $\beta'$ be the set of utility prices when $i^*$ does not participate in the market. Furthermore, set $\beta'_{i^*}$ equal to the value that uses the whole budget of buyer $i^*$ under the paced auction allocation using $\beta'_{-i^*}$.
  Let $\beta$ be the AMO-EG solution when $i^*$ reports truthfully.

  Let $\mathcal{I}$ be the set of all buyers whose types are in some subset $\Theta' \subseteq \Theta$ where $\theta_{i^*}\in \Theta'$. Let $p_{\Theta'}$ be the probability of sampling a buyer from $\Theta'$. The expected utility that $i^*$ gains from misreporting for a given market size with $n$ buyers can be decomposed as follows:
  \[
    p\left\{|I| \geq \frac{np_{\Theta'}}{2} \right\} \expectation\left[u_i\bigg| |I| \geq \frac{np_{\Theta'}}{2}\right]
      + p\left\{|I| < \frac{np_{\Theta'}}{2} \right\} \expectation\left[u_i\bigg||I| < \frac{np_{\Theta'}}{2}\right].
  \]
  From Lemma~\ref{lem:multinomial} we know that the second term converges exponentially quickly to zero as the market gets large, and so even when $i^*$ receives \emph{all} items the contribution from this term converges to zero. Thus we just need to show that the first term converges to zero.

  Let $\epsilon = \max_{i, i' \in \mathcal{I}, j \in J} |v_{ij} - v_{i'j}|$.
  Now consider any $i \in \mathcal{I}$. We wish to bound $\beta_i$ in terms of $\beta_{i^*}$. The multiplicative bound \eqref{eq:beta mult bound} applies because $\epsilon$ satisfies the condition of Lemma~\ref{lem:beta bounds}.
  We thus have for all $i\in \mathcal{I}$
  \begin{align}
    \label{eq:beta sandwich spl}
    \epsilon^{\downarrow}\beta_{i^*} \leq \beta_i \leq \epsilon^{\uparrow} \beta_{i^*}.
  \end{align}

  Now let $\alpha = \frac{\beta_{i^*}}{\beta_{i^*}'}$, by monotonicity $\alpha \geq 1$. Applying \eqref{eq:beta sandwich spl} to $\beta$ we then have $\beta_i \geq \epsilon^{\downarrow}\beta_{i^*} \geq \alpha \epsilon^{\downarrow} \beta_{i^*}'$. The total increase in payments (where payments include $p$ and $\lambda$) can thus be lower bounded based on increases over $\mathcal{I}$ as follows
  \begin{align}
    \label{eq:ic payment increases}
    B_i = 1 = \sum_{i,j} x_{ij} (p_j + \lambda_{ij}) - \sum_{i,j} x_{ij}' (p_j' + \lambda_{ij}')
    &\geq
    \sum_{i \in \mathcal{I},j} x_{ij} (\beta_i v_{ij} -  \beta_i' v_{ij})\\
    &\geq
    \sum_{i \in \mathcal{I},j} x_{ij} (\alpha\epsilon^{\downarrow}\beta_{i^*}' v_{ij} -  \epsilon^{\uparrow}\beta_{i^*}' v_{ij})\\
    &=
    \beta_{i^*}' \sum_{i \in \mathcal{I},j} x_{ij} v_{ij} (\alpha\epsilon^{\downarrow} -  \epsilon^{\uparrow} )\\
    & \geq 
    (\alpha\epsilon^{\downarrow} -  \epsilon^{\uparrow} ) \beta_{i^*}' |\mathcal{I}| \frac{kNv^{\downarrow}}{N} \\
    & =
    (\alpha\epsilon^{\downarrow} -  \epsilon^{\uparrow} ) \beta_{i^*}' |\mathcal{I}| kv^{\downarrow}.
  \end{align}
  The first inequality follows because prices are monotonically increasing, and we know that for any $j$ s.t. $x_{ij}>0$ for some $i\in \mathcal{I}$ we have that either 1) $i$ is winning at least $x_{ij}$ of $j$ under $\beta$ in which case the price increase on $x_{ij}$ is exactly as stated, or $i$ is not winning $x_{ij}$ under $\beta_i$, in which case $\beta_iv_{ij}\leq p_j$, and so someone else is paying at least $\beta_iv_{ij}$ for the fraction $x_{ij}$ of $j$.
  The second inequality follows from \eqref{eq:beta sandwich}.
  The third inequality follows from our utility lower bound \eqref{eq:u lb} adapted to the low rank AMO setting.

  It follows that the the utility price increase $\alpha$ can be bounded as
  \[
     \alpha\epsilon^{\downarrow} -  \epsilon^{\uparrow} \leq \frac{1}{\beta_{i^*}' |\mathcal{I}| k v^{\downarrow} }.
     \Leftrightarrow
     \alpha
     \leq \frac{1}{\epsilon^{\downarrow}\beta_{i^*}' |\mathcal{I}| k v^{\downarrow} } +  \frac{\epsilon^{\uparrow}}{\epsilon^{\downarrow}}.
  \]
  We know that $\beta_{i^*}$ is bounded from below independent of market size as per Lemma~\ref{lem:beta ub}. 

  For any particular $\Theta'$, the expected size of the corresponding set $\mathcal{I}$ grows linearly in market size, and thus the likelihood that a constant fraction of the expected size is not reached decreases exponentially in market size, as shown in Lemma~\ref{lem:multinomial}.
It follows that for any $\alpha > \frac{\epsilon^{\uparrow}}{\epsilon^{\downarrow}}$ it gets driven to $\alpha \leq \frac{\epsilon^{\uparrow}}{\epsilon^{\downarrow}}$ as the market gets large.
  But we are free to pick our bounding low-rank set $\Theta'$ as small as we wish around $\theta_i$, and thus we can make $\frac{\epsilon^{\uparrow}}{\epsilon^{\downarrow}}$ arbitrarily small by picking a sufficiently-small $\Theta'$.

The above shows that in the case where $|I| \geq \frac{np_{\Theta'}}{2}$, as the market gets large $i^*$ cannot improve their utility by manipulating prices and then buying according to an optimal single utility price that decides their bids. It remains to show that when $|I| \geq \frac{np_{\Theta'}}{2}$ the difference between the utility of buying at an optimal single utility price, versus manipulating individual valuations such that $i^*$ can buy at prices $p$, is small.
But we already know from Corollary~\ref{cor:regret bound} that each buyer $i\in \mathcal{I}, i \ne i^*$ has bounded regret under $\beta'$, since $|I| \geq \frac{np_{\Theta'}}{2}$ guarantees that the corollary applies for  $n > \frac{2s}{np_{\Theta'}}$. Furthermore, since the utilities $v_i, v_{i^*}$ satisfy $\|v_i - v_{i^*}\|_{\infty} \leq \epsilon$ and the $\ell_1$ length of the allocation vectors $x_i,x_{i^*}$ are bounded by $\frac{u^{\uparrow}_i}{v^{\downarrow}}$ by Lemma~\ref{lem:bounded total allocation}, we know that buyer $i^*$ prefers their best allocation under $\beta'$ over that of buyer $i$ by at most $2\epsilon \frac{u^{\uparrow}_i}{v^{\downarrow}}$. Combining these bounds we find that under the prices from $\beta'$, buyer $i^*$ can at most improve their utility by buying at prices $p$, as opposed to via a single utility-price, by
\[
  2\epsilon \frac{u^{\uparrow}_i}{v^{\downarrow}}
  + \frac{u_i^{\uparrow}}{v^{\downarrow}}\max_k \beta_i^{-1} \left(  \eps_{\beta}v_{ik} + \beta_{i'}\epsilon\right).
\]
Since $\epsilon$ and $\epsilon_{\beta}$ both get smaller as we make $\Theta'$ smaller we get that this gain can be made arbitrarily small by picking $\Theta'$ sufficiently small, and picking a correspondingly-large $n$ such that $|I| \geq \frac{np_{\Theta'}}{2}$.
\end{proof}

\end{document}